\documentclass[letterpaper,10 pt, conference]{ieeeconf}
\IEEEoverridecommandlockouts                               
\usepackage{fancyhdr}
\usepackage{hyperref}
\usepackage{xspace}
\usepackage[font={small}]{caption}
\usepackage{color}
\usepackage{setspace}
\usepackage{enumerate}
\usepackage{graphics}
\usepackage{cite}
\usepackage{dsfont}
\usepackage{latexsym}
\usepackage{amsmath}
\usepackage{amssymb}
\usepackage{amsfonts}
\usepackage{amsthm}
\usepackage{times}
\usepackage{sgame}
\usepackage{color}
\usepackage{verbatim}
\usepackage{xcolor}
\usepackage{pgfplots}
\usepackage{epsfig} 
\usepackage{subcaption}
\usepackage{hyperref}

\usepackage{multirow,array}

\let\labelindent\relax
\usepackage{enumitem}
\newlist{inparaenum}{enumerate}{2}% allow two levels of nesting in an enumerate-like environment
\setlist[inparaenum]{nosep}% compact spacing for all nesting levels
\setlist[inparaenum,1]{label=\bfseries\arabic*.}% labels for top level
\setlist[inparaenum,2]{label=\arabic{inparaenumi}\emph{\alph*})}% labels for second level
\def\mbb{\mathbb}
\def\mcal{\mathcal}

\def\dSP0{\delta_{SP0}}
\def\dRT0{\delta_{RT0}}
\def\dPS1{\delta_{PS1}}
\def\dTR1{\delta_{TR1}}
\def\half{\frac{1}{2}}
\def\prt{\partial}
\def\dgdx{\frac{\partial g_1}{\partial x}}
\def\dgdn{\frac{\partial g_1}{\partial n}}

\def\int{ {\text{int}} }

\def\balph{\bar{\alpha}}

\def \be {\begin{equation}}
\def \ee {\end{equation}}
\def \ba {\begin{aligned}}
\def \ea {\end{aligned}}

\theoremstyle{plain}% default
\newtheorem{theorem}{Theorem}[section]

\newtheorem{proposition}{Proposition}[section]

\newtheorem{assumption}{Assumption}

\def\bs{\boldsymbol}

\normalsize\title{\LARGE \bf
	Two competing populations with a common environmental resource
	\thanks{ }}

\author{
	Keith Paarporn, James Nelson
	\thanks{K. Paarporn and J. Nelson are with the Department of Computer Science, University of Colorado, Colorado Springs. Contact: \texttt{ \{kpaarpor,jnelso22\}@uccs.edu}. This work is supported in part by the UCCS Committee on Research and Creative Works, and NSF grant  \#ECCS-2346791.
	}
}

\begin{document}
\thispagestyle{plain}
\pagestyle{plain}

\maketitle

\begin{abstract}
 	Feedback-evolving games is a framework that models the co-evolution between payoff functions and an environmental state. It serves as a useful tool to analyze many social dilemmas such as natural resource consumption, behaviors in epidemics, and the evolution of biological populations. However, it has primarily focused on the dynamics of a single population of agents. In this paper, we consider the impact of two populations of agents that share a common environmental resource. We focus on a scenario where individuals in one population are governed by an environmentally ``responsible" incentive policy, and individuals in the other population are environmentally ``irresponsible". An analysis on the asymptotic stability of the coupled system is provided, and conditions for which the resource collapses are identified. We then derive consumption rates for the irresponsible population that optimally exploit the environmental resource, and analyze how incentives should be allocated to the responsible population that most effectively promote the environment via a sensitivity analysis.

\end{abstract}

%%%%%%%%%%%%%%%%%%%%%%%%%%%%%%%%%%%%%%%%%%%%%%%%%%%%%%%
\section{Introduction}

Many game-theoretic analyses describe strategic interactions between individuals provided that the incentives for their choices are static, i.e. do not change over time. However, choices often have an impact on a shared environment, which in turn affect their incentives for future choices. The utilization of common resources such as water, fishing grounds, or traffic networks best illustrates this -- high individual utilization makes fewer resources available to others in the future \cite{ostrom1990governing}. 

%Clearly, a dynamic interplay exists between users' decisions and their impact on a shared environmental state.

The recent framework of \emph{feedback-evolving games} addresses this interplay by incorporating a changing environmental state coupled with existing evolutionary game theoretic dynamics \cite{weitz2016oscillating,tilman2020evolutionary,gong2022limit}. The core model considers payoff functions to individuals that depend on whether the environment is in an abundant or depleted state. It is primarily used to understand how the payoff functions should be structured in order to avoid a ``tragedy of the commons", an outcome in which the environment ultimately becomes depleted. In other words, these payoff functions reflect environmental incentive policies that can be designed to manage the common resource \cite{chen2018punishment,paarporn2018optimal,wang2020steering}.

Much of the existing research focuses on a \emph{single population} of agents whose decisions impact the environment in isolation \cite{satapathi2023coupled,frieswijk2023modeling,khazaei2021disease,wang2020steering,arefin2021imitation,stella2021impact,stella2022lower}. These models are not sufficient to describe multi-population interactions. For example, individuals in neighboring countries follow different environmental policies, yet utilize resources from the same common source (e.g. fish in the ocean, water from rivers, clean air). Extensions of feedback-evolving games to multiple populations, and possibly multiple local environments, would enable a vastly richer set of scenarios for study. Recent works in the literature have shifted attention to these multi-population scenarios \cite{kawano2018evolutionary,govaert2022population,certorio2023epidemic}. For instance, \cite{certorio2023epidemic} studies epidemic spreading within and between two populations, where the behaviors in each population have externalities on the health state of the other. Overall, an interesting and understudied direction involves hierarchical decision-making, i.e. the strategic selection of local incentive policies given there are environmental externalities between the populations.

In this paper, we extend the framework in this direction by considering two populations that share the same local environmental resource. We focus on a particular setting where one population, labelled ``responsible", implements a pro-environmental incentive policy such that the common resource can be sustained in the absence of other populations. The other population, labelled ``irresponsible", does not restrain its consumption activity. Our study centers on the following question, stated informally as
\begin{center}
	\emph{To what extent can the irresponsible population exploit the environmental resource?}
\end{center}
While too much consumption could cause the resource to collapse, too little consumption leaves missed opportunities. Our highlighted contributions are:
\begin{itemize}
 	\item A stability analysis for the dynamics of the two-population system.
	\item Identify conditions for which the environment collapses, and when it can be sustained.
	\item Derivation of the optimal consumption rates for the irresponsible population.
	\item A sensitivity analysis for the incentive policies chosen by the responsible population.
\end{itemize}
The last item above illustrates how incentives should be allocated in order to maximally promote the environmental resource. Interestingly, we find that incentivizing mutual cooperation (i.e. coordinated) is far more effective in promoting resource levels than incentivizing unilateral cooperation (anti-coordinated). 

%subsidizing unilateral cooperation

The paper is organized as follows. Section \ref{sec:model} provides relevant background on single population feedback-evolving games, before presenting the two-population model. Section \ref{sec:dynamics} states our main result regarding the dynamics of the model. Section \ref{sec:utility} proposes and solves an optimization problem regarding the irresponsible population's choice of consumption levels. Simulations and sensitivity analyses are given in Section \ref{sec:sensitivity}, followed by concluding remarks.

%%%%%%%%%%%%%%%%%%%%%%%%%%%%%%%%%%%%%%%%%%%%%%%%%%%%%%%
\section{Model}\label{sec:model}

Before describing our two-population model, preliminary background on single-population feedback-evolving games is provided.

\subsection{Preliminary on Feedback-evolving games}

A single-population feedback-evolving game describes a population of agents that have access to a degradable environmental resource. The relative abundance of the resource is denoted by $n \in [0,1]$. At any given time, each agent is choosing whether to use a low consumption action (strategy $\mcal{L}$), or a high consumption action (strategy $\mcal{H}$). High consumption degrades $n$, and low consumption improves $n$. The immediate payoff experienced by an agent is described by the environment-dependent $2\times 2$ payoff matrix,
\begin{equation}
	A_n = n\begin{bmatrix} R_1 & S_1 \\ T_1 & P_1 \end{bmatrix} + (1-n)\begin{bmatrix} R_0 & S_0 \\ T_0 & P_0 \end{bmatrix}
\end{equation}
Here, the first row and column corresponds to a low consumer, and the second row and column corresponds to a high consumer. Entry $ij$ ($i,j \in \{\mcal{L},\mcal{H}\}$) indicates the experienced payoff to an agent using strategy $i$ when encountering an agent using strategy $j$. We denote $x \in [0,1]$ as the fraction (or frequency) of agents in the population using strategy $\mcal{L}$. The payoff experienced by each type of agent is then given by
\be
	\pi_\mcal{L}(x,n) =  [A_n [x,1-x]^\top ]_1, \quad \pi_\mcal{H}(x,n) = [A_n [x,1-x]^\top ]_2
\ee
The payoffs are determined by the parameters in the $A_0$ and $A_1$ matrices. The $A_1$ matrix is the payoff matrix when the environment is abundant. Following the literature on feedback-evolving games, we make the following assumption about the $A_1$ matrix. 
\begin{figure}[t]
	\centering
	\includegraphics[scale=.20]{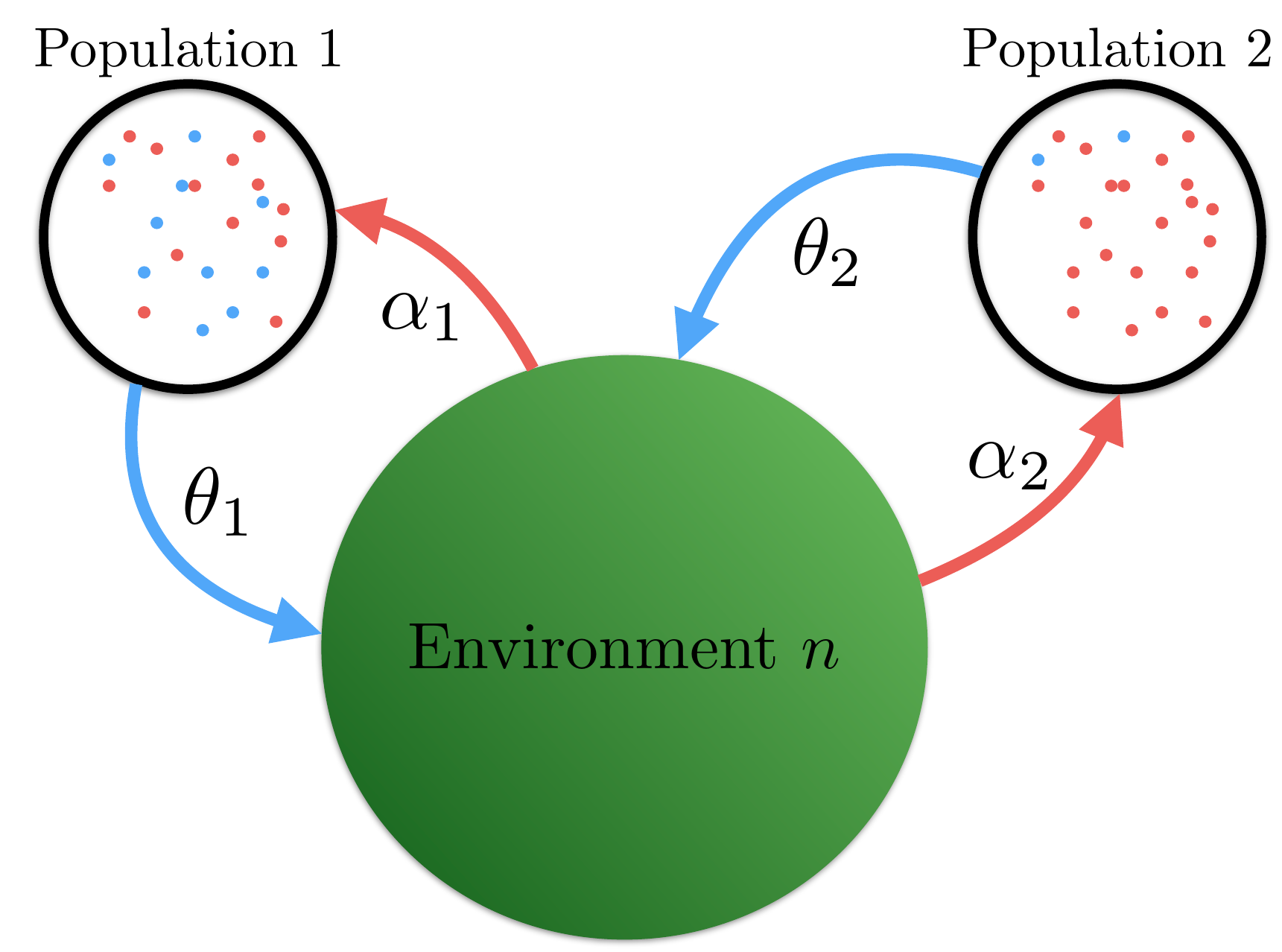}
	\caption{Diagram of our two-population model. Activities from both populations impact the shared environmental state, $n$. The rates $\theta_i$, $\alpha_i$ denote the restoration and degradation rates from the activities of population $i$, respectively.}
	\label{fig:main}
\end{figure}
\begin{assumption}\label{assume:PD}
	High consumption is the dominant strategy in $A_1$, i.e. $\dTR1 := T_1 - R_1 > 0$ and $\dPS1 := P_1 - S_1 > 0$
\end{assumption}
On the other hand, the $A_0$ matrix describes payoffs when the environment is depleted. We interpret $A_0$ to be an ``\emph{environmental policy}" that the population follows. For example, the low consumption strategy may be more incentivized when the environment is bad (e.g. subsidies for using electric vehicles). The payoff structure of the $A_0$ matrix is completely determined from the parameters $\dSP0 := S_0 - P_0$ and $\dRT0 := R_0 - T_0$.

We will use the replicator equation to describe how agents revise their decisions over time. Moreover, the environment $n$ evolves over time, as it is influenced by the decisions in the population. The overall system dynamics is given by two coupled ODEs:
\begin{equation}\label{eq:1pop}
	\begin{aligned}
		\dot{x} &= x(1-x)g(x,n) \\
		\dot{n} &= \epsilon n(1-n)(\theta x - \alpha(1-x))
	\end{aligned}
\end{equation}
where
\be
	g(x,n) := \pi_{\mcal{L}}(x,n) - \pi_{\mcal{H}}(x,n)
\ee
is the payoff difference between low and high consumers, $\theta > 0$ is the restoration rate from low consumption activity, $\alpha > 0$ is the degradation rate from high consumption activity, and $\epsilon > 0$ is a time-scale separation constant. The form of the $\dot{n}$ equation is often referred to as the \emph{tipping point} dynamics, since $n$ is increasing only if there are sufficiently high fraction of low consumers.

We consider initial conditions in the interior $(x,n) \in (0,1)^2$, which is forward-invariant under \eqref{eq:1pop}. Any initial condition on an edge will stay on that edge. The originating work \cite{weitz2016oscillating} provided a full characterization of the asymptotic outcomes of system \eqref{eq:1pop} for all possible environmental policy matrices $A_0$. 

The result below summarizes the variety of behaviors that system \eqref{eq:1pop} can exhibit.

\begin{theorem}[adapted from \cite{weitz2016oscillating}]\label{thm:PNAS}
	The environmental policy $(\dSP0,\dRT0)$ determines the asymptotic properties of \eqref{eq:1pop} as follows.
	
	\vspace{1mm}
	
	\noindent 1) \emph{Sustained resource:} If  $(\dSP0,\dRT0) \in \mcal{V}$, where
	\be
		\mcal{V} := \{ (y_1,y_2) \in \mbb{R}^2 : y_1 > 0 \text{ and } -\frac{\theta}{\alpha} y_1 < y_2 < \frac{\dTR1}{\dPS1} y_1 \},
	\ee
	then the fixed point $(x^*,n^*) = (\frac{\alpha}{\alpha + \theta}, \frac{g(x^*,0)}{-\prt g/\prt n(x^*)}) \in (0,1)^2$ is the only asymptotically stable fixed point in the system.
%	\be\label{eq:g_coeffs}
%	\ba
%	a &:= \dSP0 - \dRT0 + \dPS1 - \dTR1 \\
%	b &:= \dRT0 - \dSP0 \\
%	c &:= -(\dPS1 + \dSP0) \\
%	d &:= \dSP0.
%	\ea
%	\ee
	\vspace{1mm}
	
	\noindent 2) \emph{Oscillating Tragedy of the commons (OTOC):} If $\dSP0 > 0$ and $\frac{\dTR1}{\dPS1} \dSP0 < \dRT0$, then system \eqref{eq:1pop} exhibits a stable heteroclinic cycle between the four corner fixed points, $(0,0)$, $(1,0)$, $(1,1)$, $(0,1)$.
	
	\vspace{1mm}
	
	\noindent 3) \emph{Tragedy of the commons (TOC):} If  $\dSP0 \leq 0$, or  $\dSP0 > 0$ and $\dRT0 < -\frac{\theta}{\alpha} \dSP0$, then the only asymptotically stable fixed point has $n = 0$.
\end{theorem}
The  policies belonging to $\mcal{V}$ are ones that can sustain a stable and non-zero resource level. In other words, they avert a TOC, which we refer to as a fixed point where the resource level is zero. The policies in $\mcal{V}$ can be considered as more desirable operating conditions than the policies described in items 2 and 3 above. The OTOC is considered an undesirable outcome, as the system cycles between periods of nearly zero and fully abundant resource levels. Lastly, we note that the time-scale separation $\epsilon$ does not affect the derived stability properties and fixed points of the system.

%Visually, the state trajectory ``spirals outward" from the interior, and its $\omega$-limit set is the boundary of the state space.

%%%%%%%%%%%%%%%%%%%%%%%%%%%%%%%%%%%%%%%%%%%%%%%%%%%%%%%

\subsection{Model: Two-population feedback-evolving games}

Now, we consider two populations that share the environmental resource. Figure \ref{fig:main} illustrates this scenario. Specifically, the consumption decisions of the members of both populations have effects on the environmental resource. Now, $x_i \in [0,1]$ denotes the fraction of low consumers in population $i \in \{1,2\}$. The behavior of population $i$ is governed by the payoff matrices $A_0^{(i)}$ (parameters $\dSP0^i$, $\dRT0^i$) and $A_1^{(i)}$ (parameters $\dTR1^i$, $\dPS1^i$), restoration rate $\theta_i > 0$, degradation rate $\alpha_i > 0$. In line with Assumption \ref{assume:PD}, we maintain the parameters in the abundant state satisfy $\dTR1^i, \dPS1^i > 0$. The payoff differences for each population are given by
\be
	g_i(x_i,n) := \pi_\mcal{L}^{(i)}(x_i,n) - \pi_\mcal{H}^{(i)}(x_i,n)
\ee
where $\pi_j^{(i)}$ is the experienced payoff to a $j$-strategist in population $i$. The system dynamics are now given by the set of three ODEs with state variable $\bs{z} = (x_1,x_2,n)$:

%We will extend the notation of \eqref{eq:g_coeffs} to define variables $a_i$, $b_i$, $c_i$, and $d_i$ with population-dependent parameters (e.g. $b_i = \dRT0^i - \dSP0^i$, etc.), and denote 
%\be
%	g_i(x_i,n) := a_i x_i n + b_i x_i + c_i n + d_i.
%\ee
%as the payoff difference between low and high consumers in population $i$. T

%
\begin{equation}\label{eq:2pop}
	\begin{aligned}
		\dot{x}_1 &= x_1(1-x_1)g_1(x_1,n) \\
		\dot{x}_2 &= x_2(1-x_2)g_2(x_2,n) \\
		\dot{n} &= \epsilon n(1-n)h(x_1,x_2)
	\end{aligned}
\end{equation}
where $h(x_1,x_2) := \sum_{i=1}^2 (\theta_i x_i - \alpha_i (1-x_i) )$, and with initial condition $(x_1(0), x_2(0), n(0)) \in (0,1)^3$. By construction, $(0,1)^3$ is forward-invariant. We are interested in studying how the dynamics of the two-population environmental coupled system \eqref{eq:2pop} behaves with respect to the population's defining parameters -- especially how the environmental policies $A_0^{(1)}$, $A_0^{(2)}$ adopted by each population affects environmental resources. We will make the following assumptions on these environmental policies.
\begin{assumption}\label{assume:dgdn}
	We assume that $\frac{\prt g_i}{\prt n}(x_i) < 0$ for all $x_i \in [0,1]$, $i = 1,2$. This is equivalent to $\dSP0^i > -\dPS1^i$ and $\dRT0^i > -\dTR1^i$.
\end{assumption}
Assumption \ref{assume:dgdn} asserts that the relative payoff to low consumers in both populations monotonically decreases as the environmental state improves. The next and final assumption focuses our study on a scenario where one population 1 is environmentally \emph{responsible}, and population 2 is \emph{irresponsible}.
\begin{assumption}\label{assume:responsible}
	For population 1, we assume that $(\dSP0^1,\dRT0^1) \in \mcal{V}$ (with $\theta = \theta_1$ and $\alpha = \alpha_1$). For population 2, we assume that $\dSP0^2, \dRT0^2 < 0$.
\end{assumption}
The first part of Assumption \ref{assume:responsible} asserts that agents in population 1 are cooperative enough to sustain the resource (Theorem \ref{thm:PNAS}, item 1) in the absence of population 2. The second part asserts that agents in population 2 make no effort to conserve resources,  always preferring the high consumption strategy.  The set of all feasible policies specified by Assumptions  \ref{assume:dgdn} and \ref{assume:responsible}  is visually depicted in Figure \ref{fig:assumption}.

\begin{figure}[t]
	\centering
	\includegraphics[scale=.26]{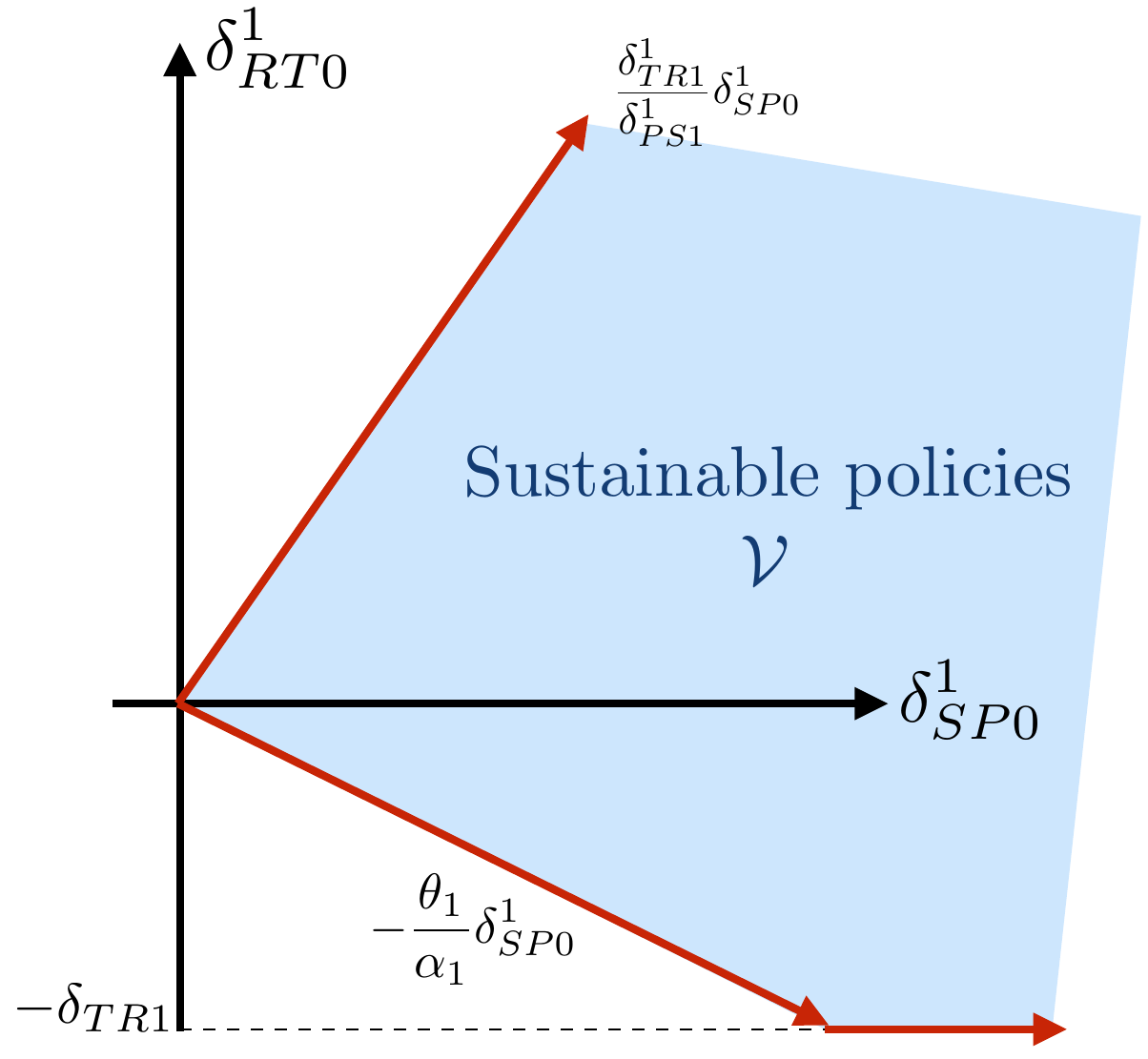}
	\caption{The set of all sustainable policies for population 1 is shown as the blue region. A sustainable policy maintains a stable, nonzero resource in the absence of other populations (Theorem \ref{thm:PNAS}). This region is the focus of Assumptions \ref{assume:dgdn} and \ref{assume:responsible}.}
	\label{fig:assumption}
\end{figure}
%

%%%%%%%%%%%%%%%%%%%%%%%%%%%%%%%%%%%%%%%%%%%%%%%%%%%%%%%

\section{Stability analysis of two-population game}\label{sec:dynamics}

In this section, we characterize the dynamical behavior of the two-population coupled system by analyzing the local stability properties of all of its fixed points. 

\begin{theorem}\label{thm:stability}
	The asymptotic dynamics of the two-population system  \eqref{eq:2pop} are summarized below.
	
	\noindent 1) Suppose $\alpha_2 > \theta_1$. Then a tragedy of the commons is  asymptotically stable, i.e. $\lim_{t\rightarrow \infty} n(t) = 0$.
	
	\vspace{1mm}
	
	\noindent 2) Suppose $\alpha_2 < \theta_1$. 
	\begin{enumerate}[label=(\alph*)]
		\item If $\frac{\alpha_2 - \theta_1}{\alpha_1 + \alpha_2}\dSP0^1 <  \dRT0^1 < \frac{\dTR1}{\dPS1}\dSP0^1$, then the only asymptotically stable fixed point is of the form $(x_1^*,0,n^*)$, where
	\be\label{eq:nstar}
		x_1^* := \frac{\alpha_1 + \alpha_2}{\alpha_1 + \theta_1}, \quad n^* := -\frac{g_1(x_1^*,0)}{\dgdn(x_1^*)}.
	\ee
		\item If $\dRT0^1 \leq \frac{\alpha_2 - \theta_1}{\alpha_1 + \alpha_2}\dSP0^1$, then a tragedy of the commons is the only stable outcome.
	\end{enumerate}
	
	\vspace{1mm}
	
	\noindent 3) Suppose $\alpha_2 = \theta_1$. If $\dRT0^1 > 0$, then there is a locally stable line segment of fixed points given by 
	\be
		\left\{ (1,0,n) : n \in \left[0,\frac{\dRT0^1}{\dRT0^1 + \dTR1^1}\right) \right\}
	\ee
	All other fixed points are isolated and unstable. If $\dRT0^1 \leq 0$, then a tragedy of the commons is asymptotically stable.
\end{theorem}
Several remarks are in order. In item 1 above, the irresponsible population induces a tragedy of the commons if its consumption rate is higher than the responsible population's restoration rate. Item 2a  provides a region of sustainable policies for population 1. This region gets smaller as $\alpha_2$ increases while remaining less than $\theta_1$. Item 2b provides a region where the population 1 policy fails to sustain the resource even though $\alpha_2 < \theta_1$. Item 3 presents an edge case where there are an infinite number of fixed points, all with different resource levels. The fixed point on the line that the system converges to depends on the initial condition.

Moreover, we observe the sustained resource level $n^*$ in item 2a is a decreasing function in $\alpha_2$.  Taking the derivative with respect to $\alpha_2$, we obtain
\be\nonumber
	\frac{\prt n^*}{\prt \alpha_2} = \frac{g_1(x_1^*,0) \frac{\prt^2 g_1}{\prt n \prt x} - \dgdx(x_1^*,0)\dgdn(x_1^*) }{(\alpha_1+\theta_1)(\dgdn(x_1^*))^2}
%	\frac{\prt n^*}{\prt \alpha_2} = -\frac{b_1c_1 - a_1d_1}{(\alpha_1+\theta_1) (\dgdn(\bar{\alpha}_1+ \bar{\alpha}_2))^2} < 0.
\ee
It is negative since the numerator being negative is equivalent to the condition $\dRT0^1 \leq \frac{\dTR1}{\dPS1}\dSP0^1$.

The proofs for each part of Theorem \ref{thm:stability} are presented in Appendix \ref{sec:app_dynamics}.

%The sustained level $n^* \in (0,1)$. To verify this, we note $g_1$ is  a bilinear function in $x_1$ and $n$:
%\be
%	g_1(x_1,n) = a_1 x_1 n + b_1 x_1 + c_1 n + d_1
%\ee
%where
%	\be\label{eq:g_coeffs}
%	\ba
%	a_1 &:= \dSP0^1 - \dRT0^1 + \dPS1^1 - \dTR1^1 \\
%	b_1 &:= \dRT0^1 - \dSP0^1 \\
%	c_1 &:= -(\dPS1^1 + \dSP0^1) \\
%	d_1 &:= \dSP0^1.
%	\ea
%	\ee
%We may then express
%\be
%	\ba
%	n^* &= -\frac{b_1 (\bar{\alpha}_1+ \bar{\alpha}_2) + d_1}{a_1(\bar{\alpha}_1+ \bar{\alpha}_2) + c_1} 
%	\ea
%\ee
%where $\bar{\alpha}_j := \frac{\alpha_j}{\alpha_1 + \theta_1}$. To show $n^*>0$, note that the denominator is negative by Assumption \ref{assume:dgdn}. The numerator is positive as it can be written $\dRT0^1 (\bar{\alpha}_1+ \bar{\alpha}_2) + \dSP0^1 \frac{\theta_1 - \alpha_2}{\alpha_1 + \alpha_2} $, and this quantity is positive if and only if $\frac{\alpha_2 - \theta_1}{\alpha_1 + \alpha_2}\dSP0^1 \leq  \dRT0^1$. Thus, $n^* > 0$. The condition $n^* < 1$ is equivalent to $(\alpha_2 - \theta_1) \dPS1^1 - (\alpha_1 + \alpha_2) \dTR1^1 < 0$, which is satisfied for $\alpha_2 \leq \theta_1$.

%%%%%%%%%%%%%%%%%%%%%%%%%%%%%%%%%%%%%%%%%%%%%%%%%%%%%%%

%
\begin{figure*}
	\centering
	\includegraphics[scale=.3]{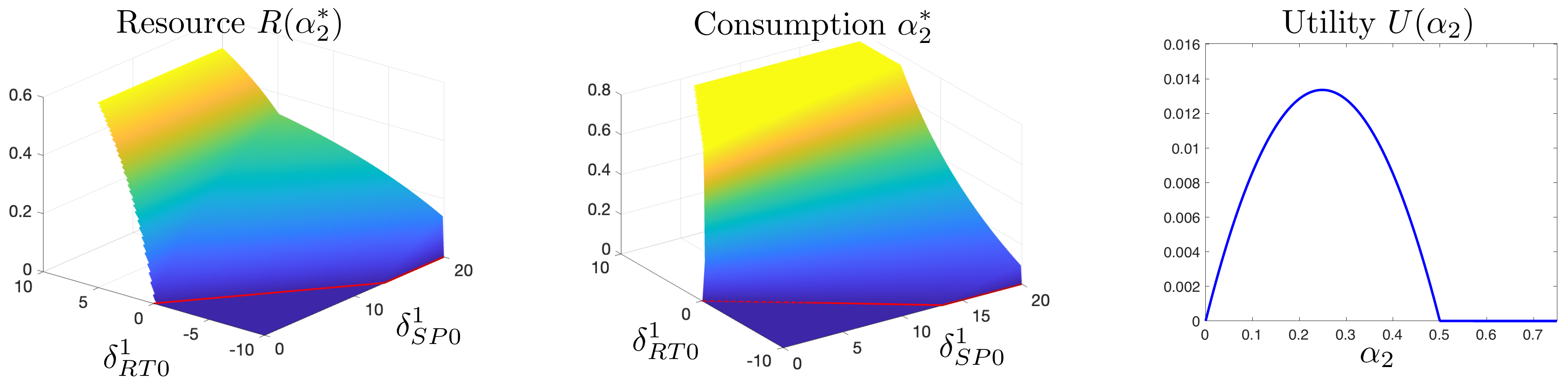}
	\caption{(Left) This surface plot shows the resource level that results from the optimal consumption rate $\alpha_2^*$ detailed in Theorem \ref{thm:consume}. The red line indicates the bottom border of the set of feasible policies for population 1. We observe that $R(\alpha_2^*)$ is increasing in both $\dRT0^1$ and $\dSP0^1$.  (Center) The optimal consumption rate detailed in Theorem \ref{thm:consume}.  (Right) An example of the utility function $U_2(\alpha_2)$ under policy $(\dSP0^1,\dRT0^1) = (3,-0.5)$. For all simulations, the parameter values are: $\dTR1 = 10$, $\dPS1 = 6$, $\theta_1 = 0.75$, $\alpha_1 = 1$, $\epsilon = 0.1$.}
	\label{fig:sims}
\end{figure*}

\section{Results: exploitation of resources}\label{sec:utility}

In this section, we study the following hierarchical decision problem posed informally as: \emph{how much consumption can the irresponsible population get away with?} To approach this question, we consider a local authority for population 2 that may set the consumption rate $\alpha_2 \geq 0$. This decision is representative of, for example, water usage or fishing regulations. The \emph{utility} that the authority seeks to maximize is defined by
\be\label{eq:U}
	U(\alpha_2) := \alpha_2 \cdot R(\alpha_2).
\ee
where $R(\cdot)$ summarizes the resource level that results in the asymptotic outcome of the dynamics. From Theorem \ref{thm:stability}, we define it as
\be\label{eq:R}
	R(\alpha_2) :=
	\begin{cases}
		n^*(\alpha_2), &\text{if } \alpha_2 \leq \theta_1 \text{ and} \\ 
		&\quad \frac{\alpha_2 - \theta_1}{\alpha_1+\alpha_2} \dSP0^1 \leq \dRT0^1 \leq \frac{\dTR1}{\dPS1}\dSP0^1 \\
		0, &\text{else }
	\end{cases}
\ee
where $n^*(\alpha_2) := -\frac{g_1(\frac{\alpha_1 + \alpha_2}{\alpha_1 + \theta_1},0)}{\dgdn(\frac{\alpha_1 + \alpha_2}{\alpha_1 + \theta_1})}$ is taken directly from \eqref{eq:nstar}. In the edge case $\alpha_2 = \theta_1$, Theorem \ref{thm:stability} states there is a range $n \in [0,\frac{\dRT0^1}{\dRT0^1 + \dTR1^1}]$ of possible asymptotic resource levels depending on the initial condition. In our definition of $R(\cdot)$, we have elected to assign the highest possible resource\footnote{This is done for two reasons. First, it makes the resource function well-defined and left-continuous at $\theta_1$. Thus, $R$ attains a maximum value in the interval $[0,\theta_1]$. Second, since no particular initial condition is prescribed in the model, we may view the irresponsible population to be ``optimistic'' regarding the best case among all possible outcomes.} value $R(\alpha_2) = \frac{\dRT0^1}{\dRT0^1 + \dTR1^1}$.

The choice to increase consumption comes at the cost of worsening or even destroying the environmental resource. Thus, \eqref{eq:U} captures the tension between resource consumption and the stability of the resource. The optimal consumption rate can be determined by solving the  problem

\be\label{eq:OC}
	\alpha_2^* := \arg\max_{\alpha_2 \geq 0} U(\alpha_2). \tag{OC}
\ee
Note here that we are assuming a fixed environmental policy $(\dSP0^1,\dRT0^1)$ for population 1 that satisfies Assumption \ref{assume:responsible}. Our main result below  provides a full characterization of the optimal consumption rate and utility.

\begin{theorem}[Optimal consumption rate]\label{thm:consume}
	The optimal solution and value of \eqref{eq:OC} are given as follows.
	\begin{enumerate}[label=(\alph*)]
		\item If $C(\dSP0^1) \leq \dRT0^1 < \frac{\dTR1^1}{\dPS1^1} \dSP0^1$, then
		\be
			\alpha_2^* = \theta_1
%			U^* = -\theta_1\frac{g_1(1,0)}{\dgdn(1)}
		\ee
		and $R(\alpha_2^*) = \frac{\dRT0^1}{\dRT0^1 + \dTR1^1}$.
		
%		$R(\alpha_2^*) = -\frac{g_1(1,0)}{\dgdn(1)}$.
		\item If $\max\{-\frac{\theta_1}{\alpha_1} \dSP0^1, -\dTR1^1 \} \leq \dRT0^1 \leq C(\dSP0^1)$, then
		\be\label{eq:crit_alpha}
			\alpha_2^* = -\frac{(\alpha_1 + \theta_1)}{a_1}\dgdn(\bar{\alpha}_1) \left( 1 - \sqrt{\frac{Y}{b_1 \dgdn(\bar{\alpha}_1) }} \right)
%			U^* = a_1^{-2} (\alpha_1+\theta_1)\left(\sqrt{b_1 \dgdn(\bar{\alpha}_1)} - \sqrt{T} \right)^2
		\ee
		and $R(\alpha_2^*) = \frac{\sqrt{b_1\dgdn(\bar{\alpha}_1)} - \sqrt{Y} }{a_1 \sqrt{b_1^{-1}\dgdn(\bar{\alpha}_1)}}$, where we have defined
		\be
			\small
		\ba
			&C(\dSP0^1) := \half\left[ -((1-\bar{\alpha}_1) \dPS1^1 + \dTR1^1 ) + \right. \\
			&\quad\quad \left. \sqrt{((1-\bar{\alpha}_1)\dPS1^1 + \dTR1^1 )^2 + 4(1-\bar{\alpha}_1) \dTR1^1 \dSP0^1 } \right]
		\ea
		\ee
		\be
		\ba
			a_1 &:= \dSP0^1 - \dRT0^1 + \dPS1^1 - \dTR1^1 \\
			b_1 &:= \dRT0^1 - \dSP0^1 \\
			\bar{\alpha}_1 &:= \frac{\alpha_1}{\alpha_1 + \theta_1}		
			\ea\nonumber
		\ee
		and
		\be
			Y := \dTR1^1 \dSP0^1 - \dRT0^1 \dPS1^1
		\ee
	\end{enumerate}
\end{theorem}
Item a) specifies the range of population 1 environmental policies where population 2 benefits most from the maximal consumption rate $\alpha_2^* = \theta_1$. Any higher consumption rate will cause the resource to collapse. In this range, the resource (and consequently, utility) is highly sensitive at the threshold, since $R(\theta_1) > 0$ and discontinuously drops to zero for $\alpha_2 > \theta_1$. 

Item b) gives the range of environmental policies where population 2 benefits most from a consumption rate that is not maximal, i.e. $\alpha_2^* < \theta_1$ \eqref{eq:crit_alpha}. Any higher consumption causes the resource to degrade marginally faster, i.e. $U(\alpha_2)$ becomes a decreasing function for $\alpha_2 > \alpha_2^*$. Unlike in the policies from item a), the utility $U(\alpha_2)$ maintains continuity on $\alpha_2 \geq 0$, even as it becomes zero for all $\alpha_2 > \theta_1$.

The proof of Theorem \ref{thm:consume} is given in Appendix \ref{sec:app_consumption}. Numerical computations of the optimal quantities are shown in Figure \ref{fig:sims}.

%%%%%%%%%%%%%%%%%%%%%%%%%%%%%%%%%%%%%%%%%%%%%%%%%%%%%%%
\section{Simulations: Sensitivity of cooperation incentives}\label{sec:sensitivity}

From the perspective of population 2, the optimal consumption rate is chosen to always result in a non-zero resource level, $R(\alpha_2^*)$, given a \emph{fixed} population 1 policy $(\dSP0^1,\dRT0^1)$. In this section, we seek to understand how the choice of environmental policy $(\dSP0^1,\dRT0^1)$ impacts the resulting resource level, under the optimal consumption rate for population 2. Intuitively, suppose a local authority for population 1 has the option to administer incentives to promote the low consumption strategy, given that the irresponsible population will optimally exploit the updated policy. 

The updated policy becomes  $(\dSP0^1 + u_s,\dRT0^1 + u_r)$, where $u_s, u_r \geq 0$ are the added incentives. The payoff matrix $A_0^1$ then reads as
\be
	\begin{bmatrix} R_0 + u_r & S_0 + u_s \\ T_0 & P_0 \end{bmatrix}
\ee
The addition of the $u_r$ term represents added incentives for \emph{mutual cooperation} -- individuals that practice low consumption are rewarded when many also practice the low consumption strategy. On the other hand, the addition of the $u_s$ term represents added incentives for \emph{unilateral cooperation} -- individuals that practice low consumption are rewarded when many practice the high consumption strategy. How should the authority select $u_s,u_r$? Which type of incentive is more effective in promoting the resource level?

\subsection{Calculation of sensitivities}

Our approach here is to evaluate the sensitivity of $R(\alpha_2^*)$ with respect to small changes in the policy. In other words, we calculate the gradient of the (re-defined) function $R^*(\dSP0^1,\dRT0^1) := R(\alpha_2^*)$ with respect to $(\dSP0^1,\dRT0^1)$. 

The partial derivative of $R^*$ with respect to $\dSP0^1$, in the region $C(\dSP0^1) \leq \dRT0^1 \leq \frac{\dTR1^1}{\dPS1^1} \dSP0^1$ (part (a) of Theorem \ref{thm:consume}), is simply $\frac{\prt R^*}{\prt \dSP0^1} = 0$. In this case, the addition of any $u_s$ incentive has \emph{no effect} on the resource level. In the region $\max\{-\frac{\theta_1}{\alpha_1} \dSP0^1, -\dTR1^1 \} \leq \dRT0^1 \leq C(\dSP0^1)$ (part (b) of Theorem \ref{thm:consume}), the partial derivative with respect to $\dSP0^1$ is calculated to be
\be
	\ba
	\frac{\prt R^*}{\prt \dSP0^1} &= \frac{\dPS1^1-\dTR1^1}{a_1^2}(1 - \phi) + \\ 
	&\frac{-b_1}{2a_1} \phi\left(\frac{1}{-b_1} - \frac{\dTR1^1}{Y} + \frac{1-\balph_1}{-\dgdn(\balph_1)} \right)
	\ea
\ee
where
\be
	\phi := \sqrt{\frac{Y}{b_1 \dgdn(\balph_1)}},
\ee 
depends on the policy $(\dSP0^1,\dRT0^1)$, and $a_1$, $b_1$, and $Y$ are variables that also depend on the policy $(\dSP0^1,\dRT0^1)$, and were defined in the statement of Theorem \ref{thm:consume}. The sign of $b_1=\dRT0^1-\dSP0^1$ is negative, since $C(\dSP0^1) \leq \dSP0^1$ with equality if and only if $\dSP0^1 = 0$.

The partial derivative of $R^*$ with respect to $\dRT0^1$, in the region $C(\dSP0^1) \leq \dRT0^1 \leq \frac{\dTR1^1}{\dPS1^1} \dSP0^1$ (part (a) of Theorem \ref{thm:consume}), is calculated to be
\be\nonumber
	\frac{\prt R^*}{\prt \dRT0^1} = \frac{\dTR1^1}{(\dRT0^1 + \dTR1^1)^2} > 0.
\ee
In the region $\max\{-\frac{\theta_1}{\alpha_1} \dSP0^1, -\dTR1^1 \} \leq \dRT0^1 \leq C(\dSP0^1)$ (part (b) of Theorem \ref{thm:consume}), the partial derivative with respect to $\dRT0^1$ is calculated to be
\be
	\ba
	\frac{\prt R^*}{\prt \dRT0^1} &= \frac{\dTR1^1 - \dPS1^1}{a_1^2}(1 - \phi) + \\ 
	&\frac{-b_1}{2a_1} \phi\left(\frac{1}{b_1} + \frac{\dPS1^1}{Y} + \frac{\balph_1}{-\dgdn(\balph_1)} \right)
	\ea\nonumber
\ee

\begin{figure*}
	\centering
	\includegraphics[scale=.3]{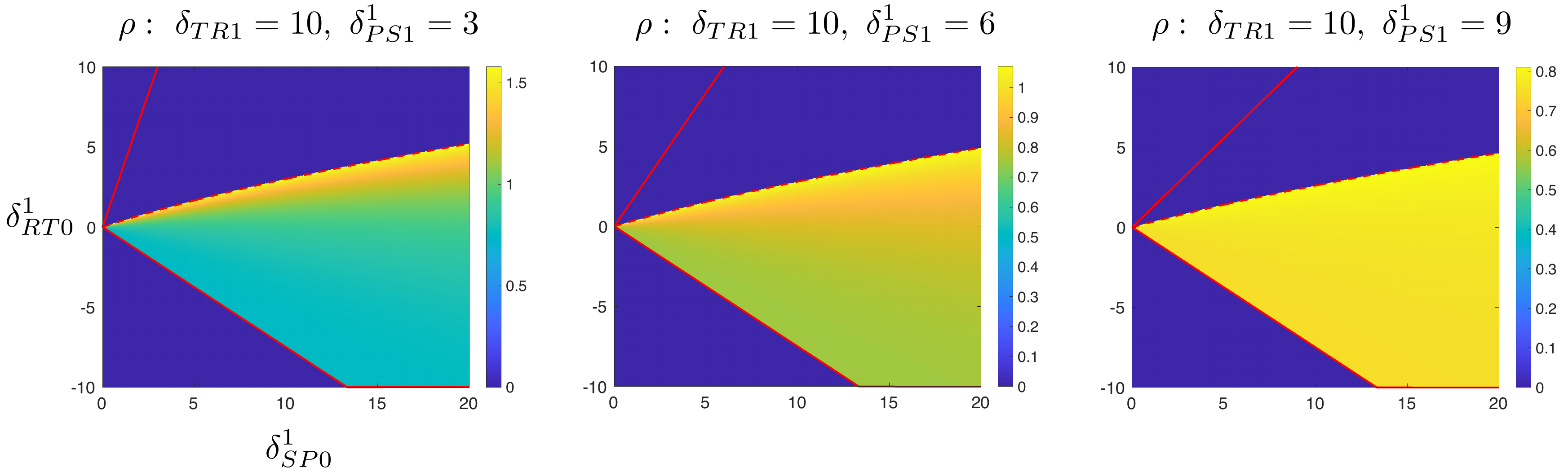}
	\caption{\emph{Sensitivity ratios}. This series of plots shows the sensitivity ratio $\rho$ as the population 1 policy $(\dSP0^1,\dRT0^1)$ varies. The red lines indicate the boundaries of the set of feasible policies for population 1. In all plots, we observe the ratio is highest near the curve $C(\dSP0^1)$ (dashed red line), i.e. $u_s$ becomes effective, and indeed can be more effective than $u_r$ for lower values of $\dPS1^1$ (left and center plots). However,  $u_s$ is generally not as effective relative to $u_r$. For higher values of $\dPS1^1$, $u_s$ will never be as effective as $u_r$, i.e. $\rho < 1$ for all policies (right plot). For all simulations, the parameter values are: $\dTR1 = 10$, $\theta_1 = 0.75$, $\alpha_1 = 1$, $\epsilon = 0.1$.}
	\label{fig:sensitivity}
\end{figure*}

The following basic property holds for both sensitivities.
\begin{proposition}
	For any set of fixed parameter values $\dTR1^1$, $\dPS1^1$, $\bar{\alpha}_1$, and any feasible policy $(\dSP0^1,\dRT0^1)$, it holds that $\frac{\prt R^*}{\prt \dSP0^1},\frac{\prt R^*}{\prt \dRT0^1} \geq 0$.
\end{proposition}
Thus, increasing incentive for low consumption can never make the resource level worse under the optimal consumption rate of population 2.

\subsection{Simulations: comparison of incentives}

To address the question of which type of incentive, $u_r$ or $u_s$, is more effective, we inspect the \emph{sensitivity ratio}
\be
	\rho(\dSP0^1,\dRT0^1) := \frac{\prt R^*/\prt \dSP0^1}{\prt R^*/ \prt \dRT0^1}
\ee
A ratio of $\rho < 1$ indicates that the incentive $u_r$ is more effective than $u_s$, and $\rho > 1$ indicates the opposite. A numerical computation of the sensitivity ratio is provided in Figure \ref{fig:sensitivity}. 

While a formal analysis is not yet provided in this paper, the numerical computations strongly suggest that the $u_r$ incentive is generally more effective than $u_s$. We observe that $\rho < 1$ for a large space of policies $(\dSP0^1,\dRT0^1)$, and $\rho > 1$ only when $(\dSP0^1,\dRT0^1)$ is close to the threshold curve $(\dSP0^1,C(\dSP0^1))$ and $\dPS1^1$ is sufficiently small relative to $\dTR1^1$. For larger values of $\dPS1^1$, we observe that $\rho > 1$ for all policies $(\dSP0^1,\dRT0^1)$. We also reiterate that $\rho = 0$ for all policies above the threshold curve $C(\dSP0^1)$, i.e. only the $u_r$ incentive can help improve the resource level.

%%%%%%%%%%%%%%%%%%%%%%%%%%%%%%%%%%%%%%%%%%%%%%%%%%%%%%%
\section{Conclusion and Future Work}

In this paper, we studied a feedback-evolving game with two populations that share a common environmental resource. This marks initial steps in extending the framework to multi-population interactions and hierarchical decsion-making. We focused on a scenario where one population is ``responsible" about using environmental resources, while the other population is ``irresponsible". We evaluated to what extent the irresponsible population can take advantage of the resources by deriving optimal consumption rates. Lastly, a sensitivity analysis was provided that suggests incentivizing mutual cooperation is more effective than unilateral cooperation. Future work will analyze how multiple irresponsible populations interact with one another. Beyond the scope of the present study, one future direction is to investigate populations that each have local environmental resources. Another interesting consideration is to study evolutionary dynamics beyond the replicator equation.

%This work sets the stage for future studies involving hierarchical decision-making with multiple populations. For example, how would multiple irresponsible populations interact with one another? One can also consider scenarios where populations have their own local environmental resources that have externalities on others.

\bibliographystyle{IEEEtran}
\bibliography{library}

%%%%%%%%%%%%%%%%%%%%%%%%%%%%%%%%%%%%%%%%%%%%%%%%%
%=== ================APPENDIX ============================================
%%%%%%%%%%%%%%%%%%%%%%%%%%%%%%%%%%%%%%%%%%%%%%%%%
\appendices

\appendix

\begin{table*}
    \begin{center}
        \begin{tabular}{|c|c|c|c|}
        \hline
        Fixed point & Interpretation & Existence condition & Stability condition \\
        \hline
        $\bs{z}_\text{A} := (\frac{\alpha_1 + \alpha_2}{\alpha_1 + \theta_1},0,-\frac{b_1 x_1 + d_1}{a_1 x_1 + c_1})$ & Sustained environment & $\alpha_2 < \theta_1$, $\frac{\alpha_2 - \theta_1}{\alpha_1 + \alpha_2}\dSP0^1 <  \dRT0^1$ & $\dRT0^1 \leq \frac{\dTR1}{\dPS1}\dSP0^1$ \\        
        \hline
        $\bs{z}_\text{B} := (\frac{\dSP0^1}{\dSP0^1 - \dRT0^1},0,0)$ & Tragedy, mixed consumption & $\dRT0^1 \leq 0$  & $\dRT0^1 \leq \frac{\alpha_2 - \theta_1}{\alpha_1+\alpha_2}\dSP0^1$ \\
        \hline
        $\bs{z}_\text{C1} := (0,0,0)$ & Tragedy & Always   & Never (Assumption 3) \\
        \hline
        $\bs{z}_\text{C2} :=(0,0,1)$ & Abundance & Always  & Never \\
        \hline
        $\bs{z}_\text{C3} :=(1,0,0)$ & Tragedy &  Always &  $\dRT0^1 \geq 0$, $\alpha_2 > \theta_1$ \\
        \hline
        $\bs{z}_\text{C4} :=(1,0,1)$ & Abundance & Always  & Never \\
        \hline
        $(1,0,n)$ & Line segment $n\in[0,1]$ & $\alpha_2 = \theta_1$, $\dRT0^1 \geq 0$  &   $n \in [0,\frac{\dRT0^1}{\dRT0^1 + \dTR1^1} )$ \\
        \hline
        $(\frac{\dPS1^1}{\dPS1^1 - \dTR1^1},0,1)$ & Abundance & Never  & N/A \\
        \hline
        \end{tabular}
    \end{center}
    \caption{This table accompanies the proof of Theorem \ref{thm:stability}. It lists all fixed points of the form $(x_1,0,n)$. The column ``Existence condition'' specifies the paramter regime where the fixed point is contained in the set $[0,1]^3$. The column ``Stability condition'' specifies parameter regime where the fixed point is locally asymptotically stable, or equivalently when all of its eigenvalues are negative. The exception is the second-to-last row, which indicates a line of fixed points along a one-dimensional center manifold. Here, the condition indicates when its other two eigenvalues are negative.}
    \label{table:FPs}
\end{table*}

Throughout the Appendix, we make use of the following notations. The payoff difference function $g_i(x_i,n)$ ($i=1,2$) is bilinear in its two arguments, which we can write as
\be\nonumber
	g_i(x_i,n) = a_i x_i n + b_i x_i + c_i n + d_i
\ee
where we define
	\be\label{eq:g_coeffs}
	\ba
	a_i &:= \dSP0^i - \dRT0^i + \dPS1^i - \dTR1^i \\
	b_i &:= \dRT0^i - \dSP0^i \\
	c_i &:= -(\dPS1^i + \dSP0^i) \\
	d_i &:= \dSP0^i.
	\ea
	\ee
Consequently, we can write $\frac{\prt g_i}{\prt n}(x_i) = a_i x_i + c_i < 0$ (Assumption \ref{assume:dgdn}) and $\frac{\prt g_i}{\prt x_i}(n) = a_i n + b_i$. Moreover, we will make use of the variable
\be
	Y := b_1c_1 - a_1 d_1 = \dTR1^1\dSP0^1 - \dRT0^1\dPS1^1
\ee
which is positive if and only if $\dRT0^1 < \frac{\dTR1}{\dPS1}\dSP0^1$, and zero if it holds with equality.

\subsection{Analysis of dynamical system}\label{sec:app_dynamics}

To establish stability properties of system \eqref{eq:2pop}, we need to analyze its Jacobian matrix, which we compute below:
\be
	J(\bs{z}) := 
	\begin{bmatrix}
		J_{11} & 0 & J_{13}  \\
		0 & J_{22} & J_{23} \\
		\epsilon J_{31} & \epsilon J_{32} & \epsilon J_{33}  \\
	\end{bmatrix}
\ee
for any $\bs{z} \in [0,1]^3$, where
\be
	\ba
		J_{11} &= x_1(1-x_1)\frac{\prt g_1}{\prt x_1}(n) + (1-2x_1)g_1(x_1,n)  \\
		J_{13} &= x_1(1-x_1) \frac{\prt g_1}{\prt n}(x_1) \\
		J_{22} &= x_2(1-x_2)\frac{\prt g_2}{\prt x_2}(n) + (1-2x_2)g_2(x_2,n) \\
		J_{23} &= x_2(1-x_2)\frac{\prt g_2}{\prt n}(x_2)\\
	\ea\nonumber
\ee
\be
	\ba
		J_{31} &= n(1-n)(\theta_1 + \alpha_1) \\
		J_{32} &= n(1-n)(\theta_2 + \alpha_2) \\
		J_{33} &= (1-2n) (\theta_1 x_1 + \theta_2 x_2 - \alpha_1 (1-x_1) - \alpha_2 (1-x_2))
	\ea\nonumber
\ee

%A summary of all such fixed points and their stability conditions are given in Table \ref{table:FPs}, assuming that Assumptions \ref{assume:PD}, \ref{assume:dgdn}, and \ref{assume:responsible} hold.

% Line UB \frac{\dRT0^1}{\dRT0^1 + \dTR1^1}

\begin{proof}[Proof of Theorem \ref{thm:stability}]
	We first observe that by Assumption \ref{assume:responsible}, it must be the case that $x_2(t) \rightarrow 0$ for any initial condition in $(0,1)^3$ because $g_2(x_2,n) < 0$ for all $x_2,n$. In other words, for any $\zeta > 0$  there is some finite time $t_\zeta > 0$ such that $x_2(t) < \zeta$ for all $t \geq t_\zeta$. Hence, we only need to analyze stability properties of fixed points of the form $(x_1,0,n)$, i.e. ones lying in the $x_1,n$ plane. A complete list of the fixed points of this form and their stability properties is given in Table \ref{table:FPs}. The stability properties were derived by evaluating the eigenvalues of the Jacobian for each fixed point. The summary given in Table \ref{table:FPs} establishes the result of Theorem \ref{thm:stability}. However, we highlight the methodology from each item of the Theorem statement.
	
	\noindent \textbf{(1)} ${\alpha_2 > \theta_1}$.  We can analyze the limiting behavior by considering the function $V(\bs{z}) := n$. Its time derivative is  
	\be
		\epsilon n(1-n)(x_1\theta_1 - \alpha_2 - \alpha_1(1-x_1) + x_2(\theta_1 + \alpha_1) )
	\ee
	We can choose $\zeta > 0$ small enough and wait at least until time $t_\zeta$ such that $\dot{V} < 0$ for any $x_1 \in (0,1)$. By LaSalle's invariance principle, the set $n=0$ is  asymptotically stable from any interior initial condition.

%	Any fixed point that constitutes a tragedy of the commons is of the form $(x_1,0,0)$ with $x_1 \in [0,1]$. 
%	
%	For $x_1 = 1$, the eigenvalues of $J$ are $\{\dRT0^1, \dSP0^2, \epsilon(\theta - \alpha_2) \}$. They are all negative if and only if $\dRT0^1 < 0$ and $\alpha_2 > \theta_1$. 
%	
%	For $x_1 \in (0,1)$, it must follow that $g_1(x_1,0) = 0$, from which we get $x_1 = \frac{\dSP0^1}{\dSP0^1 - \dRT0^1}$. This implies $\dRT0^1 < 0$. The eigenvalues of $J$ are $\{x_1(1-x_1)(\dRT0^1 - \dSP0^1), \dSP0^2, \epsilon(x_1(\theta_1+\alpha_1) - (\alpha_1+\alpha_2) ) \}$. The first and second are negative, and the third is also negative if and only if $\frac{\alpha_2 - \theta_1}{\alpha_1+\alpha_2}\dSP0^1 > \dRT0^1$. This is satisfied since $\alpha_2 -\theta_1 > 0$ and $\dRT0^1 < 0$.
%	
%	For $x_1 = 0$, the  eigenvalues of $J$ are $\{\dSP0^1, \dSP0^2, -\epsilon(\alpha_1+\alpha_2) \}$. Since $\dSP0^1 > 0$ by Assumption \ref{assume:responsible}, $(0,0,0)$ is unstable.
%
%	Additionally, there is a line of equilibria of the form $(1,0,n)$ with $n \in [0,1]$. For any $n \in [0,1]$, the eigenvalues of $J$ are $\{ -(a_1+c_1)n - (b_1+d_1), c_2 n + d_2, \epsilon(1-2n)(\theta_1-\alpha_2) \}$. The first is negative if and only if $n < \frac{\dRT0^1}{\dRT0^1 + \dTR1}$, the second if and only if $\dSP0^2 > \dPS1^2$ (holds by Assumption \ref{assume:dgdn}), and the third if and only if $n < 1/2$. Thus, all equilibria $(1,0,n)$ with $n \in [0,\min\{\frac{\dRT0^1}{\dRT0^1 + \dTR1},\frac{1}{2} \})$ have negative eigenvalues.

	\noindent \textbf{(2a)} ${\alpha_2 < \theta_1}$ and $\frac{\alpha_2 - \theta_1}{\alpha_1 + \alpha_2}\dSP0^1 <  \dRT0^1 < \frac{\dTR1}{\dPS1}\dSP0^1$.
	
	There is a  fixed point $\bs{z}_\text{A}$ with $x_1 = \frac{\alpha_1 + \alpha_2}{\alpha_1 + \theta_1} \in (0,1)$ (from $h(x_1,n) = 0$) and $n= -\frac{b_1 x_1 + d_1}{a_1 x_1 + c_1}$ (from $g_1(x_1,n) = 0$). It holds that $n \in (0,1)$ if and only if $\frac{\alpha_2 - \theta_1}{\alpha_1+\alpha_2}\dSP0^1 < \dRT0^1$. The eigenvalues of $J$ at this point are given by
	\be\nonumber
		\lambda_1 = -(\dPS1^1 + \dSP0^1) n + \dSP0^2 < 0
	\ee
	and
	\be
		\ba
		\lambda_{2,3} &= \frac{1}{2}\left[ \frac{}{} (a_1n+b_1) x_1(1-x_1) \pm \right. \\
		&\left. \sqrt{(a_1n+b_1)^2 x_1^2(1-x_1)^2 - 4\epsilon K} \right]
		\ea\nonumber
	\ee
	where we denote
	\be\nonumber
		K = n(1-n)(\theta_1+\alpha_1)x_1(1-x_1)|\frac{\prt g_1}{\prt n}(x_1)| > 0.
	\ee
	The sign of the real parts of $\lambda_{2,3}$ is then determined by the sign of $a_1n + b_1 = \frac{-Y}{|\dgdn(x_1)|}$, which is negative if and only if $Y > 0$. This is satisfied precisely due to $ \dRT0^1 < \frac{\dTR1}{\dPS1}\dSP0^1$. Note that this condition is independent of the time-scale separation constant $\epsilon > 0$.
	
	\noindent \textbf{(2b)} ${\alpha_2 < \theta_1}$ and $\frac{\alpha_2 - \theta_1}{\alpha_1+\alpha_2}\dSP0^1 \geq \dRT0^1$.
	
	Along the bottom edge $n=0$, there is a fixed point $\bs{z}_\text{B}$ with $x_1 = \frac{\dSP0^1}{\dSP0^1 - \dRT0^1}$ (from $g_1(x_1,0) = 0$) that exists if and only if $\dRT0^1 < 0$. The eigenvalues of $J$ at this point are $\{x_1(1-x_1)b_1, \dSP0^2, \epsilon(x_1(\theta_1+\alpha_1) - (\alpha_1+\alpha_2))\}$. The first and second are negative, and the third is also negative if and only if $\frac{\alpha_2 - \theta_1}{\alpha_1+\alpha_2}\dSP0^1 > \dRT0^1$. Thus, it cannot be  stable for policies in $\mcal{V}$.
	
	There are four corner fixed points that always exist. The Jacobian evaluated at any one of the corners is a diagonal matrix. The only one that can be stable is $(1,0,0)$ if and only if $\alpha_2 > \theta_1$. There cannot be fixed points along the edges $x_1 = 0$ or $x_1 = 1$ whenever $\alpha_2 \neq \theta_1$.
	
	\noindent \textbf{(3)} $\alpha_2=\theta_1$. There is a line segment of fixed points corresponding to $x_1 = 1$ and $n \in [0,1]$. All possess a zero eigenvalue, where the Jacobian is singular along a center manifold spanned by the vector $[0,0,1]^\top$. The subset of points on this segment for which $n<\frac{\dRT0^1}{\dRT0^1 + \dTR1^1}$ have two negative eigenvalues. This subset is non-empty if and only if $\dRT0^1 > 0$. By standard arguments in center manifold theory (Ch. 2, \cite{perko2013differential}), this impilies the center manifold is locally attractive, i.e. for initial conditions close to the manifold converge to some point on the manifold. If $\dRT0^1 \leq 0$, then all fixed points on the line segment have positive eigenvalues and are thus unstable. The only fixed point that can be stable is $\bs{z}_\text{B}$ (second row of Table \ref{table:FPs}).
	
%	We now verify that all other fixed points of the form $(x_1,0,n)$ cannot be attractive. These fixed points must all lie on the boundary of the $(x_1,n)$ plane. We split into four disjoint sectors: 
%	\be
%		\ba
%		Q_1 &:= \{(x_1,n) : x_1 > \frac{\alpha_1+\alpha_2}{\alpha_1+\theta_1}, n < \frac{b_1 x_1 + d_1}{-(a_1x_1 + c_1)}\} \\
%		Q_2 &:= \{(x_1,n) : x_1 > \frac{\alpha_1+\alpha_2}{\alpha_1+\theta_1}, n \geq \frac{b_1 x_1 + d_1}{-(a_1x_1 + c_1)}\} \\
%		Q_3 &:= \{(x_1,n) : x_1 \leq \frac{\alpha_1+\alpha_2}{\alpha_1+\theta_1}, n \geq \frac{b_1 x_1 + d_1}{-(a_1x_1 + c_1)}\} \\
%		Q_4 &:= \{(x_1,n) : x_1 \leq \frac{\alpha_1+\alpha_2}{\alpha_1+\theta_1}, n < \frac{b_1 x_1 + d_1}{-(a_1x_1 + c_1)}\} \\
%		\ea
%	\ee
%	The vector field of the dynamics satisfy $\dot{x}_1, \dot{n} > 0$ in $Q_1$, $\dot{x}_1 \leq 0, \dot{n} > 0$ in $Q_2$, $\dot{x}_1 \leq 0, \dot{n} \leq 0$ in $Q_3$, and $\dot{x}_1 > 0, \dot{n} \leq 0$ in $Q_4$. The equalities hold only on the shared boundaries between the sectors. The unique interior fixed point lies at the single point of intersection of all four sectors. Moreover, we also observe that no point on the borders of $[0,1]^2$ can be attractive, since the vector field is oriented such that none of the four quadrants are positively invariant. This establishes the claim.
\end{proof}

\subsection{Analysis of optimal consumption}\label{sec:app_consumption}

The proof of Theorem \ref{thm:consume} relies on establishing certain properties about the utility function $U(\alpha_2)$. Let us define the \emph{support} of $U$ as
\be
	\text{supp } U := \{ \alpha_2 \geq 0 : U(\alpha_2) > 0 \}.
\ee
Throughout the proof, we will use the shorthand notations
\be
	\bar{\alpha}_1 := \frac{\alpha_1}{\alpha_1 + \theta_1}, \quad \bar{\alpha}_2 := \frac{\alpha_2}{\alpha_1 + \theta_1}
\ee
%\begin{lemma}
%	Suppose $\dRT0^1 \leq \frac{\dTR1^1}{\dPS1^1} \dSP0^1$. Then $\text{supp } U \subseteq [0,\theta_1]$ is an interval, and $U(\alpha_2)$ is concave on $\text{supp } U$. Moreover,
%		
%		\noindent (a)  If $C(\dSP0^1) \leq \dRT0^1 \leq \frac{\dTR1^1}{\dPS1^1} \dSP0^1$, then $U(\alpha_2)$ is strictly increasing on $\alpha_2 \in [0,\theta]$.
%		
%		\noindent (b)  If $-\frac{\theta_1}{\alpha_1} \dSP0^1 \leq \dRT0^1 \leq C(\dSP0^1)$, then $U(\alpha_2)$ attains its maximum at $\alpha_2^* \in (0,\theta)$, where
%		\be \label{eq:alphastar}
%			\alpha_2^* =(\alpha_1+\theta_1)|\dgdn(\bar{\alpha}_1)| a_1^\inv\left( 1 - \sqrt{\frac{T}{b_1 \dgdn(\bar{\alpha}_1) }} \right)
%		\ee
%\end{lemma}

\begin{proof}[Proof of Theorem \ref{thm:consume}]
	First, we notice that $U(\alpha_2) = 0$ for all $\alpha_2 > \theta_1$, so $U(\alpha_2) > 0$ only if $\alpha_2 \leq \theta_1$. The support of $U$ is the interval $[0,\theta_1]$ in the case that $\dRT0^1 > 0$, and is the interval $[0,\frac{\dSP0^1}{\dSP0^1 - \dRT0^1} \theta_1 +  \frac{\dRT0^1}{\dSP0^1 - \dRT0^1} \alpha_1]$ in the case that $\frac{- \theta_1}{\alpha_1} \dSP0^1 \leq \dRT0^1 \leq 0$. At any point in the support, the first derivative of $U$ is
	\be
		\ba
		U'(\alpha_2) &= R(\alpha_2) + \alpha_2 \frac{\prt R}{\prt \alpha_2} \\
		&= R(\alpha_2) - \alpha_2\frac{Y}{(\alpha_1+\theta_1)(\dgdn(\balph_1 + \balph_2))^2} \\
		&= \frac{b_1(\balph_1 + \balph_2) + d_1}{-(a_1 (\balph_1 + \balph_2)+ c_1)} - \balph_2 \frac{Y}{(a_1(\balph_1 + \balph_2)+ c_1)^2}
		\ea
	\ee
	We note that $Y >  0$ is equivalent to $\dRT0^1 < \frac{\dTR1^1}{\dPS1^1} \dSP0^1$. Evaluating the second derivative, we obtain
	\be
		\ba
		U''(\alpha_2) &= - 2\frac{Y}{(\alpha_1+\theta_1)(\dgdn(\balph_1 + \balph_2))^2} \\ 
		&\quad + \frac{\alpha_2 Y}{(\alpha_1+\theta_1)^2}\frac{(-2a_1)}{(\dgdn(\balph_1 + \balph_2))^3} \\
		&\propto -\left( 1 + \frac{\alpha_2}{\alpha_1 + \theta_1} \frac{a_1}{|\dgdn(\balph_1 + \balph_2)| } \right)
		\ea\nonumber
	\ee
	The sign of $U''$ then depends only on the factor in parentheses. This factor is monotonic in $\alpha_2 \in \text{supp } U$. Indeed, the sign of its derivative is the sign of ${-a_1^2 \balph_1 - a_1 c_1}$, which is independent of $\alpha_2$. If it is positive, then $U'' < 0 \ \forall \alpha_2 \in \text{supp } U$. If it is negative, we need to verify that $1 + (1 - \balph_1)\frac{a_1}{|\dgdn(1)| } > 0$. Indeed, this is equivalent to the inequality 
	\be
		(1-\balph_1) \frac{\dSP0^1 + \dPS1}{\dRT0^1 + \dTR1} > -\balph_1 
	\ee
	This is always satisfied because the LHS is positive due to Assumption \ref{assume:dgdn} ($\dRT0^1 > -\dTR1$), and the RHS is negative. This establishes concavity of $U$ in its interval of support.

	\noindent (a) In this sub-regime,  $\text{supp } U = [0,\theta_1]$. We need to establish that $U$ is strictly increasing. Since we know that $U'(0) = R(0) > 0$ and $U$ is concave, we just need to verify $U'(\theta_1) \geq 0$. We have
	\be
		\ba
			U'(\theta_1) &= R(\theta_1) - (1-\balph_1)  \frac{Y}{\dgdn(1)^2} \\ 
			&= \frac{b_1 + d_1}{-(a_1 + c_1)} - (1-\balph_1)  \frac{Y}{(a_1 + c_1)^2} 
		\ea\nonumber
	\ee
	Multiplying by $(a_1+c_1)^2$, we have $U'(\theta_1) \geq 0$ if and only if
		\be
			\ba
			&(\dRT0^1)^2 + ((1-\balph_1) \dPS1 + \dTR1) \dRT0^1 \\ 
			&\quad - (1-\balph_1) \dTR1 \dSP0^1 \geq 0
			\ea\nonumber
		\ee
		Observe the LHS above is quadratic in $\dRT0^1$, and has a negative and a positive root. Since we are considering $\dRT0^1 > 0$ in this regime, the condition above is satisfied if and only if $\dRT0^1$ is greater than the positive root, i.e.
		\be\nonumber
			\dRT0^1 \geq C(\dSP0^1)
		\ee

	\noindent (b) \underline{Case 1}: First we consider $0 < \dRT0^1 < C(\dSP0^1)$. This implies $\text{supp } U = [0,\theta_1]$ and $U'(\theta_1) < 0$. Therefore, $U(\alpha_2)$ attains its maximum for some $\alpha_2^* \in (0,\theta_1)$. To find $\alpha_2^*$, we need to solve $U'(\alpha_2) = 0$, which yields the quadratic equation
		\be
			\ba
			-a_1b_1 \bar{\alpha}_2^2 - 2b_1\dgdn(\bar{\alpha}_1) \bar{\alpha}_2 + g_1(\bar{\alpha}_1) \dgdn(\bar{\alpha}_1) = 0
			\ea
		\ee
		The ``minus" solution is given by
		\be
			\begin{aligned}
			\alpha_2^* &= \frac{\theta_1+\alpha_1}{-a_1b_1}\left[ b_1\dgdn(\bar{\alpha}_1)  -  \right. \\
			&\quad\left. \sqrt{ b_1^2\dgdn(\bar{\alpha}_1)^2 - a_1b_1 g_1(\bar{\alpha}_1,0)  \dgdn(\bar{\alpha}_1)   }\right]
			\end{aligned}\nonumber
		\ee
		which reduces to $\alpha_2^*$ in the Lemma statement. We deduce the ``minus" solution must correspond to the critical point by showing it is positive. The sign of $1 - \sqrt{\frac{Y}{b_1 \dgdn(\bar{\alpha}_1) }}$ coincides with the sign of $a_1$: 
		\be\nonumber
			1 - \sqrt{\frac{Y}{b_1 \dgdn(\bar{\alpha}_1) }} < 0 \iff a_1 (b_1 \bar{\alpha}_1 + d_1) < 0
		\ee
		Observe $b_1 \bar{\alpha}_1 + d_1 = g_1(\bar{\alpha}_1,0) > 0$ since $\dSP0^1, \dRT0^1 > 0$, and therefore $\alpha_2^*$ is positive. On the other hand, the ``plus" solution is negative since the sign of $1 +\sqrt{\frac{Y}{b_1 \dgdn(\bar{\alpha}_1) }}$ is the opposite of $a_1$.

		\noindent \underline{Case 2}: For  $\frac{- \theta_1}{\alpha_1} \dSP0^1 \leq \dRT0^1 \leq 0$, $\text{supp } U = [0,\frac{\dSP0^1}{\dSP0^1 - \dRT0^1} \theta_1 +  \frac{\dRT0^1}{\dSP0^1 - \dRT0^1} \alpha_1]$. We can express the utility as
		\be\nonumber
			U(\alpha_2) =
				\alpha_2 \cdot
				\begin{cases}
					-\frac{(\dRT0^1 - \dSP0^1)({\balph}_1 + {\balph}_2) + \dSP0^1}{a_1({\balph}_1 + {\balph}_2) + c_1}, &\text{if } \alpha_2 \in \text{supp } U \\
					0, &\text{else }
				\end{cases}
		\ee
		It holds that $U'(0) > 0$, $U'(\frac{\dSP0^1}{\dSP0^1 - \dRT0^1} \theta_1 +  \frac{\dRT0^1}{\dSP0^1 - \dRT0^1} \alpha_1) < 0$. Since $U$ is concave, it must attain its maximum in the interval $\alpha_2 \in (0, \frac{\dSP0^1}{\dSP0^1 - \dRT0^1} \theta_1 + \frac{\dRT0^1}{\dSP0^1 - \dRT0^1} \alpha_1 )$. The expression for this point is identical to $\alpha_2^*$ from Case 1 due to similar arguments.
\end{proof}

%\begin{proof}[Proof of Theorem \ref{thm:consume}]
%	We break the proof into the two items.
%	
%	\noindent (a) 
%	
%
%	\noindent (b) 
%\end{proof}

\end{document}